\documentclass[journal]{IEEEtran}

\usepackage{amssymb}
\usepackage{amsmath}
\usepackage{amsfonts}
\usepackage{float}
\usepackage{graphics}
\usepackage[tight,footnotesize]{subfigure}
\usepackage{rotating}
\usepackage{algorithm}
\usepackage{algorithmic}

\newtheorem{theorem}{Theorem}

\newtheorem{lemma}[theorem]{Lemma}

\input{tcilatex}
\begin{document}

\title{Decoding by Sampling: A Randomized Lattice Algorithm for Bounded Distance Decoding}

\pubid{}
\specialpapernotice{}
\author{Shuiyin~Liu,~Cong~Ling,~and~Damien~Stehl\'{e}\thanks{%
This work was presented in part at the IEEE International Symposium on Information Theory (ISIT
2010), Austin, Texas, US, June 2010. The third author was partly funded by the Australian Research
Council Discovery Project
DP0880724.}\thanks{%
S. Liu and C. Ling are with the Department of Electrical and Electronic Engineering, Imperial
College London, London SW7 2AZ, United Kingdom
(e-mail: shuiyin.liu06@imperial.ac.uk, cling@ieee.org).} \thanks{%
D. Stehl\'{e} is with CNRS, Laboratoire LIP (U.\ Lyon, CNRS, ENS de Lyon, INRIA, UCBL), 46 all\'ee d'Italie, 69364 Lyon Cedex 07, France (e-mail:
damien.stehle@ens-lyon.fr).}}




\maketitle%

\begin{abstract}%

Despite its reduced complexity, lattice reduction-aided decoding exhibits a widening gap to
maximum-likelihood (ML) performance as the dimension increases. To improve its performance, this
paper presents randomized lattice decoding based on Klein's sampling technique, which is a
randomized version of Babai's nearest plane algorithm (i.e., successive interference cancelation
(SIC)). To find the closest lattice point, Klein's algorithm is used to sample some lattice points
and the closest among those samples is chosen. Lattice reduction increases the probability of
finding the closest lattice point, and only needs to be run once during pre-processing. Further,
the sampling can operate very efficiently in parallel. The technical contribution of this paper is
two-fold: we analyze and optimize the decoding radius of sampling decoding resulting in better
error performance than Klein's original algorithm, and propose a very efficient implementation of
random rounding. Of particular interest is that a fixed gain in the decoding radius compared to
Babai's decoding can be achieved at polynomial complexity. The proposed decoder is useful for
moderate dimensions where sphere decoding becomes computationally intensive, while lattice
reduction-aided decoding starts to suffer considerable loss. Simulation results demonstrate
near-ML performance is achieved by a moderate number of samples, even if the dimension is as high
as 32.

\end{abstract}%

\section{INTRODUCTION}

Decoding for the linear multi-input multi-output (MIMO) channel is a problem
of high relevance in multi-antenna, cooperative and other multi-terminal
communication systems. The computational complexity associated with
maximum-likelihood (ML) decoding poses significant challenges for hardware
implementation. When the codebook forms a lattice, ML decoding corresponds
to solving the closest lattice vector problem (CVP). The worst-case
complexity for solving the CVP optimally for generic lattices is
non-deterministic polynomial-time (NP)-hard.
The best CVP algorithms to date are Kannan's \cite{kannan87} which has be shown to be of
complexity $n^{n/2+o(n)}$ where $n$ is the lattice dimension (see \cite{Stehle07}) and whose space
requirement is polynomial in~$n$, and the recent algorithm by Micciancio and
Voulgaris~\cite{MiVo10} which has complexity~$2^{O(n)}$ with respect to both time and space. In
digital communications, a finite subset of the lattice is used due to the power constraint. ML
decoding for a finite (or infinite) lattice can be realized efficiently by sphere decoding
\cite{Damen,viterbo,agrell}, whose
average complexity grows exponentially with $n$ for any fixed SNR \cite%
{jalden}. This limits sphere decoding to low dimensions in practical applications. The decoding
complexity is especially felt in coded systems. For instance, to decode the $%
4\times 4$ perfect code \cite{Oggier} using the 64-QAM constellation, one has to search in a
32-dimensional (real-valued) lattice; from \cite{jalden}, sphere decoding requires a complexity of
$64^{32\gamma }$ with some $\gamma \in \left( 0,1\right] $, which could be huge. Although some
fast-decodable codes have been proposed recently \cite{Biglieri09}, the decoding still relies on
sphere decoding.


Thus, we often have to resort to approximate solutions. The problem of solving CVP approximately
was first addressed by Babai in \cite{Babai}, which in essence applies zero-forcing (ZF) or
successive interference cancelation (SIC) on a reduced lattice. This technique is often referred
to as lattice-reduction-aided decoding \cite{yao,Windpassinger2}. It is known that ZF or minimum
mean square error (MMSE) detection aided by Lenstra, Lenstra and Lov\'{a}sz (LLL) reduction
achieves full diversity in uncoded MIMO fading channels \cite{Taherzadeh:IT,XiaoliMa08} and that
lattice-reduction-aided decoding has a performance gap to (infinite) lattice decoding depending on
the dimension $n$ only \cite{LingIT07}.
It was further shown in \cite{Jalden:LR-MMSE} that MMSE-based lattice-reduction aided decoding
achieves the optimal diversity and spatial multiplexing tradeoff. In \cite{wubbenMMSE}, it was
shown that Babai's decoding using MMSE can provide near-ML performance for small-size MIMO
systems. However, the analysis in \cite{LingIT07} revealed a widening gap to ML decoding. In
particular, both the worst-case bound and experimental gap for LLL reduction are exponential with
dimension $n$ (or linear with $n$ if measured in dB).


In this work, we present sampling decoding to narrow down the gap between lattice-reduction-aided
SIC and sphere decoding. We use Klein's sampling algorithm \cite{Klein}, which is a randomized
version of Babai's nearest plane algorithm (i.e., SIC). The core of Klein's algorithm is
randomized rounding which generalizes the standard rounding by not necessarily rounding to the
nearest integer. Thus far, Klein's algorithm has mostly remained a theoretic tool in the lattice
literature, while we are unaware of any experimental work for Klein's algorithm in the MIMO
literature. In this paper, we sample some lattice points by using Klein's algorithm and choose the
closest from the list of sampled lattice points. By varying the list size $K$, it enjoys a
flexible tradeoff between complexity and performance. Klein applied his algorithm to find the
closest lattice point only when it is very close to the input vector: this technique is known as
\emph{bounded-distance decoding} (BDD) in coding literature. The performance of BDD is best
captured by the \emph{correct decoding radius} (or simply decoding radius), which is defined as
the radius of a sphere centered at the lattice point within which decoding is guaranteed to be
correct\footnote{Although we do not have the restriction of being very close in this paper, there
is no guarantee of correct decoding beyond the decoding radius.}.

The technical contribution of this paper is two-fold: we analyze and optimize the performance of
sampling decoding which leads to improved error performance than the original Klein algorithm, and
propose a very efficient implementation of Klein's random rounding, resulting in reduced decoding
complexity. In particular, we show that sampling decoding can achieve any fixed gain in the
decoding radius (over Babai's decoding) at polynomial complexity. Although a fixed gain is
asymptotically vanishing with respect to the exponential proximity factor of LLL reduction, it
could be significant for the dimensions of interest in the practice of MIMO. In particular,
simulation results demonstrate that near-ML performance is achieved by a moderate number of
samples for dimension up to 32. The performance-complexity tradeoff of sampling decoding is
comparable to that of the new decoding algorithms proposed in \cite{Othman,Luzzi10} very recently.
A byproduct is that boundary errors for finite constellations can be partially compensated if we
discard the samples falling outside of the constellation.

Sampling decoding distinguishes itself from previous list-based detectors \cite{waters:chase,
Windpassinger1, Karen, choi:SIC,Shimokawa} in several ways. Firstly, the way it builds its list is
distinct. More precisely, it randomly samples lattice points with a discrete Gaussian distribution
centered at the received signal and returns the closest among them. A salient feature is that it
will sample a closer lattice point with higher probability. Hence, our sampling decoding is more
likely to find the closest lattice point than \cite{choi:SIC} where a list of candidate lattice
points is built in the vicinity of the SIC output point. Secondly, the expensive lattice reduction
is only performed once during pre-processing. In \cite{Windpassinger1}, a bank of $2n$ parallel
lattice reduction-aided detectors was used. The coset-based lattice detection scheme in
\cite{Karen}, as well as the iterative lattice reduction detection scheme \cite{Shimokawa}, also
needs lattice reduction many times. Thirdly, sampling decoding enjoys a proven gain given the list
size $K$; all previous schemes might be viewed as various heuristics apparently
without such proven gains. 
Note that list-based detectors (including our algorithm) may prove useful in
the context of incremental lattice decoding~\cite{NJD10}, as it provides a
fall-back strategy when SIC starts failing due to the variation of the
lattice.

It is worth mentioning that Klein's sampling technique is emerging as a
fundamental building block in a number of new lattice algorithms \cite%
{NgVi08, Gentry08}. Thus, our analysis and implementation may benefit those
algorithms as well.

The paper is organized as follows: Section II presents the transmission model and lattice
decoding, followed by a description of Klein's sampling algorithm in Section III. In Section IV
the fine-tuning and analysis of sampling decoding is given, and the efficient implementation and
extensions to complex-valued systems, MMSE and soft-output decoding are proposed in Section V.
Section VI evaluates the performance and complexity by computer simulation. Some concluding
remarks are offered in Section VII.

\textit{Notation:} Matrices and column vectors are denoted by upper and lowercase boldface
letters, and the transpose, inverse, pseudoinverse
of a matrix $\mathbf{B}$ by $\mathbf{B}^{T}$, $\mathbf{B}^{-1}$, and $%
\mathbf{B}^{\dagger }$, respectively. $\mathbf{I}$ is the identity matrix.
We denote $\mathbf{b}_{i}$ for the $i$-th column of matrix $\mathbf{B}$, $%
b_{i,j}$ for the entry in the $i$-th row and $j$-th column of the matrix $%
\mathbf{B}$, and $b_{i}$ for the $i$-th entry in vector $\mathbf{b}$. Vec$(%
\mathbf{B})$ stands for the column-by-column vectorization of the matrices $%
\mathbf{B}$. The inner product in the Euclidean space between vectors $%
\mathbf{u}$ and $\mathbf{v}$ is defined as $\langle \mathbf{u},\mathbf{v}%
\rangle =\mathbf{u}^{T}\mathbf{v}$, and the Euclidean length $\Vert \mathbf{u%
}\Vert =\sqrt{\langle \mathbf{u},\mathbf{u}\rangle }$. Kronecker product of
matrix $\mathbf{A}$ and $\mathbf{B}$ is written as $\mathbf{A}\otimes
\mathbf{B}$. $\lceil x\rfloor $ rounds to a closest integer, while $\lfloor
x\rfloor $ to the closest integer smaller than or equal to $x$ and $%
\left\lceil x\right\rceil $ to the closest integer larger than or equal to $x $. The
$\mathfrak{\Re }$ and $\mathfrak{\Im }$ prefixes denote the real and imaginary parts. A circularly
symmetric complex Gaussian random variable $x$ with variance $\sigma ^{2}$\ is defined as
$x\backsim \mathcal{CN}\left( 0,\sigma ^{2}\right) $. We write $\triangleq $ for equality in
definition. We use the standard asymptotic notation $%
f\left( x\right) =O\left( g\left( x\right) \right) $ when $\lim\sup_{x\rightarrow
\infty}|f(x)/g(x)| < \infty$ , $f\left( x\right) =\Omega \left( g\left( x\right) \right) $ when
$\lim\sup_{x\rightarrow \infty}|g(x)/f(x)| < \infty$, and $%
f\left( x\right) =o\left( g\left( x\right) \right) $ when $\lim\sup_{x\rightarrow
\infty}|f(x)/g(x)| =0$ . Finally, in this paper, the computational complexity is measured by the
number of arithmetic operations.


\section{LATTICE CODING AND DECODING}

Consider an $n_{T}\times n_{R}$ flat-fading MIMO system model consisting of $%
n_{T}$ transmitters and $n_{R}$ receivers%
\begin{equation}  \label{MIMO}
\mathbf{Y=HX+N},
\end{equation}%
where $\mathbf{X}\in \mathbb{C}^{n_{T}\times T}$, $\mathbf{Y}$, $\mathbf{N}%
\in \mathbb{C}^{n_{R}\times T}$ of block length $T$ denote the channel
input, output and noise, respectively, and $\mathbf{H}\in \mathbb{C}%
^{n_{R}\times n_{T}}$ is the $n_{R}\times n_{T}$ full-rank channel gain
matrix with $n_{R}\geq n_{T}$, all of its elements are i.i.d. complex
Gaussian random variables $\mathcal{CN}\left( 0,1\right) $. The entries of $%
\mathbf{N}$ are i.i.d. complex Gaussian with variance $\sigma ^{2}$ each.
The codewords $\mathbf{X}$ satisfy the average power constraint $E[\|\mathbf{%
X}\|^2_{\text{F}}/T] = 1$. Hence, the signal-to-noise ratio (SNR) at each
receive antenna is $1/\sigma^2$.

When a lattice space-time block code is employed, the codeword $\mathbf{X}$
is obtained by forming a $n_{T}\times T$ matrix from vector $\mathbf{s}\in
\mathbb{C}^{n_{T}T}$, where $\mathbf{s}$ is obtained by multiplying $%
n_{T}T\times 1$ QAM vector $\mathbf{x}$ by the $%
n_{T}T\times n_{T}T$ generator matrix $\mathbf{G}$ of the encoding lattice, i.e.,
$\mathbf{s}=\mathbf{Gx}$. By column-by-column vectorization of the matrices $\mathbf{Y}$ and
$\mathbf{N}$ in (\ref{MIMO}),
i.e., $\mathbf{y}=\text{Vec}(\mathbf{Y})$ and $\mathbf{n}=\text{Vec}(\mathbf{%
N})$, the received signal at the destination can be expressed as%
\begin{equation}
\mathbf{y=}\left( \mathbf{I}_{T}\otimes \mathbf{H}\right) \mathbf{Gx+n}\text{%
.}  \label{CodedSystem}
\end{equation}

When $T=1$ and $\mathbf{G=I}_{n_T}$, (\ref{CodedSystem}) reduces to the
model for uncoded MIMO communication $\mathbf{y}=\mathbf{Hx}+\mathbf{n}$.
Further, we can equivalently write%
\begin{equation}
\left[
\begin{array}{c}
\mathfrak{\Re }\mathbf{y} \\
\mathfrak{\Im }\mathbf{y}%
\end{array}%
\right] =\left[
\begin{array}{cc}
\mathfrak{\Re }\mathbf{H} & -\mathfrak{\Im }\mathbf{H} \\
\mathfrak{\Im }\mathbf{H} & \mathfrak{\Re }\mathbf{H}%
\end{array}%
\right] \left[
\begin{array}{c}
\mathfrak{\Re }\mathbf{x} \\
\mathfrak{\Im }\mathbf{x}%
\end{array}%
\right] +\left[
\begin{array}{c}
\mathfrak{\Re }\mathbf{n} \\
\mathfrak{\Im }\mathbf{n}%
\end{array}%
\right] ,  \label{RealModel}
\end{equation}%
which gives an equivalent $2n_{T}\times 2n_{R}$ real-valued model. We can
also obtain an equivalent $2n_{T}T\times 2n_{R}T$ real model for coded MIMO
like (\ref{RealModel}). The QAM constellations $\mathcal{C}$ can be
interpreted as the shift and scaled version of a finite subset $\mathcal{A}%
^{n_{T}}$ of the integer lattice $\mathbb{Z}^{n_{T}}$, i.e., $\mathcal{C}=a(%
\mathcal{A}^{n_{T}}+[1/2,...,1/2]^{T})$, where the factor $a$ arises from
energy normalization. For example, we have $\mathcal{A}^{n_{T}}=\{-\sqrt{M}%
/2,...,\sqrt{M}/2-1\}$ for \textit{M}-QAM signalling. 



Therefore, with scaling and shifting, we consider the canonical $n \times m$ ($%
m \geq n$) real-valued MIMO system model
\begin{equation}  \label{MIMOmodel}
\mathbf{y}=\mathbf{Bx}+\mathbf{n}
\end{equation}
where $\mathbf{B} \in \mathbb{R}^{m \times n}$, given by the real-valued
equivalent of $\left( \mathbf{I}_{T}\otimes \mathbf{H}\right) \mathbf{G}$,
can be interpreted as the basis matrix of the decoding lattice. Obviously, $%
n=2n_{T}T$ and $m= 2n_{R}T$. The data vector $\mathbf{x}$ is drawn from a
finite subset $\mathcal{A}^n$ to satisfy the power constraint.


A lattice in the $m$-dimensional Euclidean space $\mathbb{R}^{m}$ is
generated as the integer linear combination of the set of linearly
independent vectors \cite{gruber, cassels}:%
\begin{equation}
\mathcal{L\triangleq L}\left( \mathbf{B}\right) \mathbf{=}\left\{
\sum_{i=1}^{n}x_{i}\mathbf{b}_{i}\left\vert x_{i}\in \mathbb{Z}\text{, }%
i=1,\ldots n\right. \right\} ,
\end{equation}%
where $\mathbb{Z}$ is the set of integers, and $\mathbf{B=}\left[ \mathbf{b}%
_{1}\cdots \mathbf{b}_{n}\right] $ represents a basis of the lattice $%
\mathcal{L}$. In the matrix form, $\mathcal{L=}\left\{ \mathbf{Bx}\text{ : }%
\mathbf{x\in }\text{ }\mathbb{Z}^{n}\right\} $. The lattice has infinitely
many different bases other than $\mathbf{B}$. In general, a matrix $\mathbf{%
{B'}=BU}$, where $\mathbf{U}$ is an \textit{unimodular} matrix, i.e., $%
\det \mathbf{U=\pm }1$ and all elements of $\mathbf{U}$ are integers, is
also a basis of $\mathcal{L}$.

Since the vector $\mathbf{Bx}$ can be viewed as a lattice point, MIMO
decoding can be formulated as a lattice decoding problem. The ML decoder
computes
\begin{equation}
\mathbf{\hat{x}}=\arg \min_{\mathbf{x}\in \mathcal{A}^{n}}\Vert \mathbf{y}-%
\mathbf{Bx}\Vert ^{2}.  \label{ML}
\end{equation}%
which amounts to solving a closest-vector problem (CVP) in a finite subset of lattice
$\mathcal{L}$.\ ML decoding may be accomplished by the sphere decoding. However, the expected
complexity of sphere decoding is exponential for fixed SNR \cite{jalden}.

A promising approach to reducing the computational complexity of sphere
decoding is to relax the finite lattice to the infinite lattice and to solve
\begin{equation}
\mathbf{\hat{x}}=\arg \min_{\mathbf{x}\in \mathbb{Z}^{n}}\Vert \mathbf{y}-%
\mathbf{Bx}\Vert ^{2}.  \label{FML}
\end{equation}%
which could benefit from lattice reduction. This technique is sometimes referred to as infinite
lattice decoding (ILD). The downside is that the found lattice point will not necessarily be a
valid point in the constellation.

This search can be carried out more efficiently by lattice reduction-aided
decoding \cite{Windpassinger2}. The basic idea behind this is to use lattice
reduction in conjunction with traditional low-complexity decoders. With
lattice reduction, the basis $\mathbf{B}$ is transformed into a new basis
consisting of roughly orthogonal vectors
\begin{equation}
\mathbf{B^{\prime }}=\mathbf{BU}  \label{LRaided}
\end{equation}%
where $\mathbf{U}$ is a unimodular matrix. Indeed, we have the equivalent
channel model
\begin{equation*}
\mathbf{y}=\mathbf{B^{\prime }U}^{-1}\mathbf{x}+\mathbf{n}=\mathbf{B^{\prime
}x^{\prime }}+\mathbf{n},\quad \mathbf{x^{\prime }}=\mathbf{U^{-1}x}.
\end{equation*}%
Then conventional decoders (ZF or SIC) are applied on the reduced basis.
This estimate is then transformed back into $\mathbf{\hat{x}}=\mathbf{U\hat{x%
}^{\prime }}$. Since the equivalent channel is much more likely to be
well-conditioned, the effect of noise enhancement will be moderated. Again,
as the resulting estimate $\mathbf{\hat{x}}$ is not necessarily in $\mathcal{A}%
^{n}$, remapping of $\mathbf{\hat{x}}$ onto the finite lattice $\mathcal{A}%
^{n}$ is required whenever $\mathbf{\hat{x}}\notin \mathcal{A}^{n}$.

Babai pre-processed the basis with lattice reduction, then applied either
the rounding off (i.e., ZF) or nearest plane algorithm (i.e., SIC) \cite%
{Babai}. For SIC, one performs the QR decomposition $\mathbf{B}=\mathbf{QR}$%
, where $\mathbf{Q}$ has orthogonal columns and $\mathbf{R}$ is an upper
triangular matrix with positive
diagonal elements \cite{Horn}.
Multiplying (\ref{MIMOmodel}) on the left with $\mathbf{Q^{\dagger }}$ we
have
\begin{equation}
\mathbf{y}^{\prime }=\mathbf{Q^{\dagger }y}=\mathbf{Rx}+\mathbf{n}^{\prime }.
\label{SICmodel}
\end{equation}%
In SIC, the last symbol $x_{n}$ is estimated first as $\hat{x}_{n}=\lceil
y_{n}^{\prime }/r_{n,n}\rfloor $. Then the estimate is substituted to remove
the interference term in $y_{n-1}^{\prime }$ when $x_{n-1}$ is being
estimated. The procedure is continued until the first symbol is detected.
That is, we have the following recursion:
\begin{equation}
\hat{x}_{i}=\left\lceil \frac{y_{i}^{\prime }-\sum_{j=i+1}^{n}r_{i,j}\hat{x}%
_{j}}{r_{i,i}}\right\rfloor  \label{SIC}
\end{equation}%
for $i=n,n-1,...,1$.





Let $\mathbf{\hat{b}}_{1}$,...,$\mathbf{\hat{b}}%
_{n}$ be the Gram-Schmidt vectors where $\mathbf{\hat{b}}_{i}$ is the projection of
$\mathbf{b}_{i}$
orthogonal to the vector space generated by $\mathbf{b}_{1}$,...,$\mathbf{b}%
_{i-1}$. These are the vectors found by the Gram-Schmidt algorithm for orthogonalization.
Gram-Schmidt orthogonalization is closely related to QR decomposition. More precisely, one has the
relations $\mathbf{\hat{b}}_{i}=$ $r_{i,i}\cdot \mathbf{q}_{i}$, where $\mathbf{q}_{i}$ is the
$i$-th column of $\mathbf{Q}$. It is known that SIC finds the closest vector if the distance from
input vector $\mathbf{y}$ to the lattice $\mathcal{L}$ is less than half the length of the
shortest Gram-Schmidt vector. In other words, the \textit{correct decoding radius} for SIC is
given by
\begin{equation}
R_{\text{SIC}}=\frac{1}{2}\min_{1\leq i\leq n}\Vert \mathbf{\hat{b}}%
_{i}\Vert =\frac{1}{2}\min_{1\leq i\leq n}r_{i,i}.  \label{dbabai}
\end{equation}%

The \textit{proximity factor} defined in \cite{LingIT07} quantifies the worst-case loss in the
correct decoding radius relative to ILD%
\begin{equation}
F_{\text{SIC}}\triangleq \frac{R_{\text{ILD}}^{2}}{R_{\text{SIC}}^{2}},
\end{equation}%
where the correct decoding radius for ILD is $R_{\text{ILD}}=\lambda_1/2$ ($\lambda_1$ is the minimum distance,
 or the length of a shortest nonzero vector of the
lattice $\mathcal{L}$) and showed that under LLL reduction%
\begin{equation}
F_{\text{SIC}}\leq \beta ^{n-1},\quad \beta =\left( \delta -1/4\right) ^{-1} \label{PF-LLL}
\end{equation}%
where $1/4<\delta \leq 1$ is a parameter associated with LLL reduction \cite%
{LLL}. Note that the average-case gap for random bases $%
\mathbf{B}$ is smaller. Yet it was observed experimentally in \cite{Damien06,LingIT07} that the
average-case proximity (or approximation) factor remains exponential for random lattices.
Meanwhile, if one applies dual KZ reduction, then \cite{LingIT07}
\begin{equation}
F_{\text{SIC}}\leq n^{2}.  \label{PF-DKZ}
\end{equation}%
Again, the worst-case loss relative to ILD widens with $n$.

These finite proximity factors imply that lattice reduction-aided SIC is an instance of BDD. More
precisely, the $1/(2\gamma)$-BDD problem is to find the closest lattice point given that the
distance between input $\mathbf{y}$ and lattice $\mathcal{L}$ is less than $\lambda_1/(2\gamma)$.
It is easy to see that a decoding algorithm with proximity factor $F$ corresponds to
$1/(2\sqrt{F})$-BDD.



\section{SAMPLING DECODING}

Klein \cite{Klein} proposed a randomized BDD algorithm that increased the correct decoding radius
to
\begin{equation*}
R_{\text{Klein}}=k\min_{1\leq i\leq n}r_{i,i}.
\end{equation*}%
For the algorithm to be useful, the parameter $k$ should fall into the range $1/2<k<\sqrt{n/2}$;
in other regions Babai and Kannan's algorithms would be more efficient. Its complexity is $n^{k^{2}+O(1)}$ which for fixed~$k$ is polynomial in $n$ as $n\rightarrow \infty$.




In essence, Klein's algorithm is a randomized version of SIC, where standard rounding in SIC is
replaced by randomized rounding. Klein described his randomized algorithm in the recursive form.
Here, we rewrite it into the non-recursive form more familiar to the communications community. It
is
summarized by the pseudocode of the function Rand\_SIC$_{A}\left( \mathbf{y}%
^{\prime }\right) $ in Table I. We assume that the pre-processing of (\ref{SICmodel}) has been
done, hence the input $\mathbf{y}^{\prime }=\mathbf{Q^{\dagger }y}$ rather than $\mathbf{y}$. This
will reduce the complexity since we will call it many times. The important parameter $A$
determines the amount of randomness, and Klein suggested $A=\log n/\min_{i}r_{i,i}^{2}$.


\begin{table}\renewcommand{\arraystretch}{1.5}  \centering%
\caption{Pseudocode for the randomized SIC in sequential form}%
\begin{tabular}{l}
\hline\hline
\textbf{Function }Rand\_SIC$_{A}\left( \mathbf{y^{\prime }}\right) $ \\
1:\ \textbf{for }$i=n$ to $1$ do \\
2:\ \ \ $c_{i}\longleftarrow Ar_{i,i}^{2}$ \\
3: \ \ $\hat{x}_{i}\longleftarrow $ Rand\_Round$_{c_{i}}\left( ( {y}%
_{i}^{\prime }-\sum_{j=i+1}^{n}r_{i,j}\hat{x}_{j}) /r_{i,i}\right) $ \\
4:\ \textbf{end for} \\
5: \textbf{return }$\mathbf{\hat{x}}$ \\ \hline
\end{tabular}%
\label{table_iterative copy(1)}%
\end{table}%

The randomized SIC randomly samples a lattice point $\mathbf{z}$ that is close to $\mathbf{y}$. To
obtain the closest lattice point, one calls Rand\_SIC $K$ times and chooses the closest among
those lattice points returned, with a sufficiently large $K$. The function Rand\_Round$_{c}(r)$
rounds $r$ randomly to an integer $Q$
according to the following \textit{discrete Gaussian distribution} \cite%
{Klein}
\begin{equation}
P(Q=q)=e^{-c(r-q)^{2}}/s,\quad s=\sum_{q=-\infty }^{\infty }{e^{-c(r-q)^{2}}}%
.  \label{GaussDistr}
\end{equation}%
If $c$ is large, Rand\_Round reduces to standard rounding (i.e., decision is
confident); if $c$ is small, it makes a guess (i.e., decision is
unconfident).


\begin{lemma}
(\cite{Klein}) $s\leq s(c) \triangleq \sum_{i\geq 0}{%
e^{-ci^{2}}+e^{-c(1+i)^{2}}.}$
\end{lemma}

The proof of the lemma was given in \cite{Klein} and is omitted here. The
next lemma provides a lower bound on the probability that Klein's algorithm
or Rand\_SIC returns $\mathbf{z}\in \mathcal{L}$.

\begin{lemma}
(\cite{Klein}) Let $\mathbf{z}$ be a vector in $\mathcal{L}\left( \mathbf{B}%
\right) $ and $\mathbf{y}$ be a vector in $\mathbb{R}^{m}$. The probability
that Klein's algorithm or Rand\_SIC return $\mathbf{z}$ is bounded by%
\begin{equation}
P(\mathbf{z})\geq \frac{1}{\prod_{i=1}^{n}{s(Ar_{i,i}^{2})}}e^{-A\Vert
\mathbf{y}-\mathbf{z}\Vert ^{2}}.  \label{ProbRandDecode}
\end{equation}
\end{lemma}

\begin{proof}
The proof of the lemma was given in \cite{Klein} for the recursive version of Klein's algorithm.
Here, we give a more straightforward proof for
Rand\_SIC. Let $\mathbf{z=}\xi _{1}\mathbf{b}_{1}+\mathbf{\ldots }+\xi _{n}%
\mathbf{b}_{n}=\mathbf{B\xi }\in \mathcal{L},\xi _{i}\in \mathbb{Z}$ and
consider the invocation of Rand\_SIC$_{A}\left( \mathbf{y^{\prime }}\right) $%
. Using Lemma 1 and (\ref{GaussDistr}), the probability of $x_{i}=\xi _{i}$
is at least%
\begin{equation}
\begin{split}
& \frac{1}{{s(Ar_{i,i}^{2})}}e^{-Ar_{i,i}^{2}\left( ({y}_{i}^{\prime
}-\sum_{j=i+1}^{n}r_{i,j}\xi _{j})/r_{i,i}\right) {{^{2}}}} \\
=& \frac{1}{{s(Ar_{i,i}^{2})}}e^{-A\left( {y}_{i}^{\prime
}-\sum_{j=i+1}^{n}r_{i,j}\xi _{j}\right) {{^{2}}}}.
\end{split}%
\end{equation}%
By multiplying these $n$ probabilities, we obtain a lower bound on the
probability that Rand\_SIC returns $\mathbf{z}$%
\begin{equation}
\begin{split}
P\left( \mathbf{z}\right) & \geq \frac{1}{\prod_{i\leq n}{s(Ar_{i,i}^{2})}}%
e^{-A\sum_{i=1}^{n}\left( {y}_{i}^{\prime }-\sum_{j=i+1}^{n}r_{i,j}\xi
_{j}\right) ^{2}} \\
& =\frac{1}{\prod_{i\leq n}{s(Ar_{i,i}^{2})}}e^{-A\Vert \mathbf{y}^{\prime }-%
\mathbf{R\xi }\Vert ^{2}} \\
& \geq\frac{1}{\prod_{i\leq n}{s(Ar_{i,i}^{2})}}e^{-A\Vert \mathbf{y}-\mathbf{%
B\xi }\Vert ^{2}}.
\end{split}%
\end{equation}%
So the probability is as stated in Lemma 2.
\end{proof}

A salient feature of (\ref{ProbRandDecode}) is that the closest lattice point is the most likely
to be sampled. In particular, the lower bound resembles the Gaussian distribution. The closer
$\mathbf{z}$ is to $\mathbf{y}$, the more likely it will be sampled. Klein showed that when
$A=\log n/\min_{i}r_{i,i}^{2}$, the probability of returning $\mathbf{z}\in \mathcal{L}$ is
\begin{equation}
\Omega(n^{-\Vert \mathbf{y}-\mathbf{z}\Vert ^{2}/\min_{i}r_{i,i}^{2}}). \label{ProbRandDecode2}
\end{equation}%
The significance of lattice reduction can be seen here, as increasing $%
\min_{i}r_{i,i}^{2}$ will increase the probability lower bound (\ref{ProbRandDecode2}).

As lattice reduction-aided decoding normally ignores the boundary of the constellation, the
samples returned by Rand\_SIC$_{A}(\mathbf{y^{\prime }})$ come from an extended version of the
original constellation. We discard those samples that happen to lie outside the boundary of the
original constellation and choose the closest among the rest lattice points. When no lattice
points within the boundary are found, we simply remap the closest one back to the constellation by
``hard-limiting", i.e., remap $\hat{x}_i$ to one of the two boundary integers that is closer to
it.

\textit{Remark:} A natural questions is whether a randomized version of ZF exists. The answer is
yes. This can be done by applying random rounding in ZF. However, since its performance is not as
good as randomized SIC, it will not be considered in this paper.

\section{ANALYSIS AND OPTIMIZATION}

The list size $K$ is often limited in communications. Given $K$, the parameter $A$ has a profound
impact on the decoding performance, and Klein's choice $A=\log n/\min_{i}r_{i,i}^{2}$ is not
necessarily optimum. In this Section, we want to answer the following questions about randomized
lattice decoding:

\begin{itemize}
\item Given $K$, what is the optimum value of $A$?

\item Given $K$ and associated optimum $A$, how much is the gain in decoding
performance?

\item What is the limit of sampling decoding?
\end{itemize}

Indeed, there exists an optimum value of $A$ when $K$ is finite, since $%
A\rightarrow 0$ means uniform sampling of the entire lattice while $%
A\rightarrow \infty$ means Babai's algorithm. We shall present an approximate analysis of optimum
$A$ for a given $K$ in the sense of maximizing the correct decoding radius, and then estimate the
decoding gain over Babai's algorithm. The analysis is not exact since it is based on the correct
decoding radius only; nonetheless, it captures the key aspect of the decoding performance and can
serve as a useful guideline to determine the parameters in practical implementation of Klein's
algorithm.

\subsection{Optimum Parameter $A$}


We investigate the effect of parameter $A$ on the probability Rand\_SIC returns $\mathbf{z}\in
\mathcal{L}$. Let $A=\log \rho /\min_{i}r_{i,i}^{2}$, where ${\rho >1}$ (so that $A>0$). Then
$\rho $ is the parameter to be optimized. Since $c_{i}=Ar_{i,i}^{2}\geq \log
\rho $, we have the following bound for $s(c_{i})$:%
\begin{eqnarray}
s(c_{i}) &=&\sum_{i\geq 0}{e^{-c_{i}i^{2}}+e^{-c_{i}(1+i)^{2}}}  \notag \\
&\leq &\sum_{i\geq 0}\rho {^{-i^{2}}+\rho ^{-(1+i)^{2}}}  \notag \\
&=&1+2\left( {\rho }^{-1}+{\rho }^{-4}+{\rho }^{-9}+\ldots \right)  \notag \\
&<&1+2/{\rho +2\rho }^{-4}/\left( 1-{\rho }^{-5}\right) .
\end{eqnarray}%
Hence%
\begin{eqnarray}
\prod_{i=1}^{n}{s(c_{i})} &{<}&\left( \exp \left( 2/{\rho }+{2\rho }%
^{-4}/\left( 1-{\rho }^{-5}\right) \right) \right) ^{n}  \notag \\
&=&e^{\frac{2n}{\rho }(1+{g(\rho )})}\text{,}  \label{bound0}
\end{eqnarray}%
where ${g(\rho )=\rho }^{-3}/\left( 1-{\rho }^{-5}\right) $. With this
choice of parameter $A$, (\ref{ProbRandDecode}) can be bounded from below
by
\begin{equation}
P(\mathbf{z})>e^{-\frac{2n}{\rho }(1+{g(\rho )})}\cdot \rho ^{-\Vert \mathbf{%
y}-\mathbf{z}\Vert ^{2}/\min_{i}r_{i,i}^{2}}. \label{ProbRandDecode3}
\end{equation}

Now, let~$\mathbf{z}_{K}$ be a point in the lattice, with~$P(%
\mathbf{z}_{K})>1/K$. With~$K$ calls to Klein's algorithm, the
probability of missing~$\mathbf{z}_{K}$ is not larger than~$%
(1-1/K)^{K}<1/e$. By increasing the number of calls to $cK$ ($c\geq 1$ is a constant independent
of $n$), we can make this missing probability smaller than $1/e^c$. The value of $c$ could be
found by simulation, and $c=1$ is often enough. Therefore, any such lattice point~$\mathbf{z}_{K}$
will be found with probability close to one.
We assume that ${\rho }$ is not too small such that
${g(\rho )}$ is negligible. This is a rather weak condition: even $\rho\geq 2$ is sufficient. As
will be seen later, this condition is
indeed satisfied for our purpose. From (\ref{ProbRandDecode3}), we obtain%
\begin{eqnarray}
  e^{-\frac{2n}{\rho }}\cdot \rho ^{-\Vert \mathbf{y}-\mathbf{z}_{K}\Vert
^{2}/\min_{i}r_{i,i}^{2}}\approx \frac{1}{K}   \nonumber \\
  \Vert \mathbf{y}-\mathbf{z}_{K}\Vert \approx \min_{i}r_{i,i} \cdot \sqrt{\log _{\rho }\left(
Ke^{-2n/\rho }\right)}.\label{bound1}
\end{eqnarray}


The sampling decoder will find the closest vector point almost surely if the distance from input
vector $\mathbf{y}$ to the lattice is less than the right hand side of (\ref{bound1}), since the
probability of being sampled can only be higher than $1/K$. In this sense, the right hand side of
(\ref{bound1}) can be thought of as the decoding radius of the randomized BDD. We point out that
the right hand side of (\ref{bound1}) could be larger than $R_{\text{ILD}}$ when $K$ is
excessively large, but we are only interested in the case where it is small than $R_{\text{ILD}}$
for complexity reasons. In such a case,
we define the \emph{%
decoding radius} of sampling decoding as
\begin{equation}
R_{\text{Random}}(\rho)\triangleq \min_{1\leq i\leq n}r_{i,i}\sqrt{\log _{\rho }\left( Ke^{-2n/\rho
}\right)} . \label{bound}
\end{equation}%
This gives a tractable measure to optimize. The meaning of $R_{\text{Random}}(\rho)$ is that as
long as the distance from $\mathbf{y}$ to the lattice is less than $R_{\text{Random}}(\rho)$, the
randomized decoder will find the closest lattice point with high probability. It is natural that
$\rho $ is chosen to maximize the value of $R_{\text{Random}}(\rho)$ for the
best decoding performance. Let the derivative of $R_{\text{Random}}^2(\rho)$ with respect to $\rho $ be zero:%
\begin{equation}
\frac{\partial \left( R_{\text{Random}}^2(\rho) \right) }{\partial \rho }=\min_{1\leq i\leq
n}r_{i,i}^{2}\left( \frac{2n}{\rho ^{2}\log \rho }+\frac{2n}{\rho ^{2}\log ^{2}\rho }-\frac{\log
K}{\rho \log ^{2}\rho }\right) =0.
\end{equation}%
Because $\rho >1$, we have%
\begin{equation}
\log K=\frac{2n}{\rho }\log e\rho .
\end{equation}%
Consequently, the optimum $\rho $ can be determined from the following
equation%
\begin{equation}
K=\left( e\rho _{0}\right) ^{2n/\rho _{0}}.  \label{sizeK}
\end{equation}%
By substituting (\ref{sizeK}) back into (\ref{bound}), we get the optimum decoding radius
\begin{equation}
R_{\text{Random}} \triangleq R_{\text{Random}}(\rho_0) =\sqrt{\frac{2n}{\rho _{0}}}\min_{1\leq
i\leq n}r_{i,i}.  \label{drand}
\end{equation}

To further see the relation between $\rho _{0}$ and $K$, we calculate the
derivative of the function $f\left( \rho \right) \triangleq \left( e\rho
\right) ^{2n/\rho }$, $\rho >1$ with respect to $\rho $. It follows that%
\begin{eqnarray*}
\log f\left( \rho \right) &=&\frac{2n}{\rho }\log e\rho \\
\frac{\partial \left( f\left( \rho \right) \right) }{f\left( \rho \right)
\partial \rho } &=&-\frac{2n}{\rho ^{2}}\log e\rho +\frac{2n}{\rho ^{2}} \\
&=&-\frac{2n}{\rho ^{2}}\log \rho .
\end{eqnarray*}%
Hence%
\begin{eqnarray*}
\frac{\partial \left( f\left( \rho \right) \right) }{\partial \rho }
&=&-f\left( \rho \right) \frac{2n}{\rho ^{2}}\log \rho \\
&=&-\frac{2n}{\rho ^{2}}\left( e\rho \right) ^{2n/\rho }\log \rho ,\quad
\rho >1 \\
&<&0.
\end{eqnarray*}%
Therefore, $f\left( \rho \right) =\left( e\rho \right) ^{2n/\rho }$ is a monotonically decreasing
function when $\rho >1$. Then, we can check that a large value of $A$ is required for a small list
size $K$, while $A$ has to be decreased for a large list size $K$. It is easy to see that Klein's
choice of parameter $A$, i.e., $\rho =n$, is only optimum when $K\approx (en)^{2}$. If we choose
$K<(en)^{2}$ to reduce the implementation complexity, then $\rho _{0}>n$.

Fig. 1 shows the bit error rate against $\log \rho $ for decoding a $%
10\times 10$ (i.e., $n_{T}=n_{R}=10$) uncoded MIMO system with $K=20$, when $%
E_{b}/N_{0}=19$ dB. It can be derived from (\ref{sizeK}) that $\log \rho _{0}=4.27$. Simulation
results confirm the choice of the optimal $\rho $ offered by (\ref{sizeK}) with the aim of
maximizing $R_{\text{Random}}(\rho)$.


\begin{figure}[tbp]
\centering
\par
\includegraphics[scale=0.65]{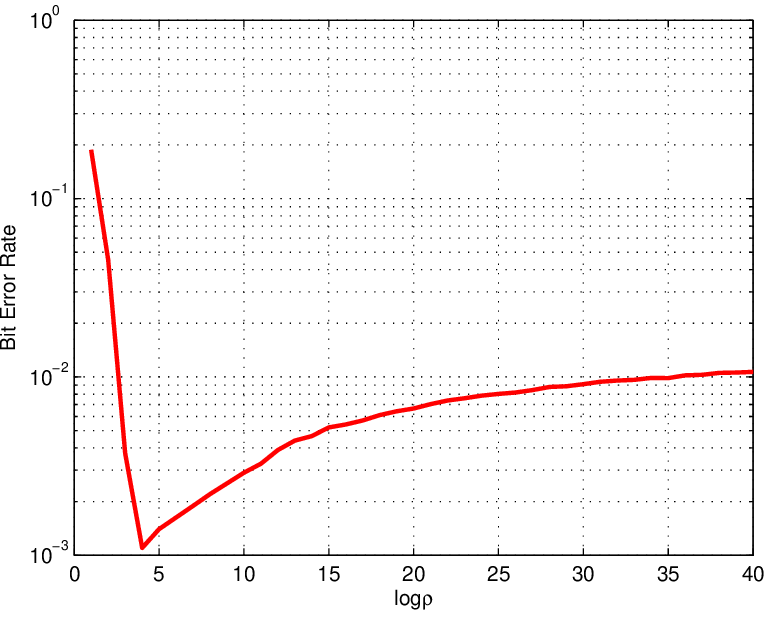}\newline
\caption{BER vs. $\log\protect\rho$ for a $10 \times 10$ uncoded system
using 64-QAM, $K=20$ and SNR per bit = 19 dB.}
\label{fig:1}
\end{figure}

\subsection{Complexity versus Performance Gain }

We shall determine the effect of complexity on the performance gain of sampling decoding over
Babai's decoding. Following \cite{LingIT07}, we define the gain in squared decoding radius as
\begin{equation*}
G\triangleq \frac{R_{\text{Random}}^{2}}{R_{\text{SIC}}^{2}}.
\end{equation*}%
%
%
%
%
%
From (\ref{dbabai}) and (\ref{drand}), we get%
\begin{equation}
G=8n/\rho _{0},\quad \rho _{0}>1.  \label{ao}
\end{equation}%
It is worth pointing out that $G$ is independent of whether or which algorithm of lattice
reduction is applied, because the term $\min_{1\leq i\leq n}r_{i,i}$ has been canceled out.

By substituting (\ref{ao}) in (\ref{sizeK}), we have%
\begin{equation}
K=\left\lceil \left( 8en/G\right) ^{G/4}\right\rceil ,\quad G<8n.
\label{G-K}
\end{equation}%
Equation (\ref{G-K}) reveals the tradeoff between $G$ and $K$. Larger $G$ requires larger $K$. For
fixed performance gain $G$, randomized lattice decoding has polynomial complexity with respect to
$n$. More precisely, each call to Rand\_SIC incurs $O(n^{2})$ complexity; for fixed $G$,
$K=O(n^{G/4})$. Thus the complexity of randomized lattice decoding is $O(n^{2+G/4})$, excluding
pre-processing (lattice reduction and QR decomposition). This is the most interesting case for
decoding applications, where practical algorithms are desired. In this case, $\rho_0$ is linear
with $n$ by (\ref{ao}), thus validating that $g(\rho)$ in (\ref{ProbRandDecode3}) is indeed
negligible.

Table \ref{Table:Gain} shows the computational complexity required to achieve the performance gain
from $3$ dB to $12$ dB. It can be seen that a significant gain over SIC can be achieved at
polynomial complexity. It is particularly easy to recover the first 3 dB loss of Babai's decoding,
which needs $O(\sqrt{n})$ samples only.

We point out that Table \ref{Table:Gain} holds in the asymptotic sense. It should be used with
caution for finite $n$, as the estimate of $G$ could be optimistic. The real gain certainly cannot
be larger than the gap to ML decoding. The closer Klein's algorithm performs to ML decoding, the
more optimistic the estimate will be. This is because the decoding radius alone does not
completely characterize the performance. Nonetheless, the estimate is quite accurate for the first
few dBs, as will be shown in simulation results.

\begin{table}\renewcommand{\arraystretch}{1.5}  \centering%
\caption{Required value of $K$ to achieve gain $G$ in randomized lattice
decoding (the complexity excludes pre-processing)}%
\begin{tabular}{|c||c|c|c|c|}
\hline Gain in dB & $G$ & $\rho _{0}$ & $K$ & Complexity \\ \hline $3$ & $2$ & $4n$ & $\sqrt{4en}$
& $O( n^{5/2}) $ \\ \hline $6$ & $4$ & $2n$ & $2en$ & $O( n^3) $ \\ \hline $9$ & $8$ & $n$ & $(
en) ^{2}$ & $O( n^{4}) $ \\ \hline
$12$ & $16$ & $n/2$ & $( en/2) ^{4}$ & $O( n^{6}) $ \\
\hline
\end{tabular}%
\label{Table:Gain}%
\end{table}%



\subsection{Limits}


Sampling decoding has its limits. Because equation (\ref{ao}) only holds when $\rho _{0}>1$, we
must have $G<8n$. In fact, our analysis requires that $\rho _{0}$ is not close to 1. Therefore, at
best sampling decoding can achieve a linear gain $G=O(n)$.
To achieve near-ML performance asymptotically, $G$ should exceed the proximity factor, i.e.,%
\begin{equation}\label{FG}
F_{\text{SIC}} \leq G=8n/\rho _{0},\quad \rho _{0}>1.
\end{equation}%
However, this cannot be satisfied asymptotically, since $F_{\text{SIC}}$ is exponential in $n$ for
LLL reduction (and is $n^2$ for dual KZ reduction). Of note is the proximity factor of random
lattice decoding $F_{\text{Random}}=F_{\text{SIC}}/G$, which is still exponential for LLL
reduction.



Further, if we do want to achieve $G>8n$, sampling decoding will not be useful. One can still
apply Klein's choice $\rho=n$, but it will be even less efficient than uniform sampling.
Therefore, at very high dimensions, sampling decoding might be worse than sphere decoding if one
sticks to ML decoding.

The $G=O(n)$ gain is asymptotically vanishing compared to the exponential proximity factor of LLL.
Even this $O(n)$ gain is mostly of theoretic interest, since $K$ will be huge. Thus, sampling is
probably best suited as a polynomial-complexity algorithm to recover a fixed amount of the gap to
ML decoding.

Nonetheless, sampling decoding is quite useful for a significant range of $n$ in practice.  On one
hand, it is known that the real gap between SIC and ML decoding is smaller than the worst-case
bounds; we can run simulations to estimate the gap, which is often less than 10 dB for $n\leq 32$.
On the other hand, the estimate of $G$ does not suffer from such worst-case bounds; thus it has
good accuracy. For such a range of $n$, sampling decoding performs favorably, as it can achieve
near-ML performance at polynomial complexity.

%

\section{IMPLEMENTATION}

In this Section, we address several issues of implementation. In particular, we propose an
efficient implementation of the sampler, extend it to complex-valued lattices, to soft output, and
to MMSE.

\subsection{Efficient Randomized Rounding}

The core of Klein's decoder is the randomized rounding with respect to discrete Gaussian
distribution (\ref{GaussDistr}). Unfortunately, it can not be generated by simply quantizing the
continuous Gaussian distribution. A rejection algorithm is given in \cite{Devroye} to generate a
random variable with the discrete Gaussian distribution from the continuous Gaussian distribution;
however, it is efficient only when the variance is large. From (\ref{GaussDistr}), the variance in
our problem is less than $1/\log \rho_0$. From the analysis in Section IV, we recognize that $\rho
_{0}$ can be large, especially for small $K$. Therefore, the implementation complexity can be
high.


Here, we propose an efficient implementation of random rounding by
truncating the discrete Gaussian distribution and prove the accuracy of this
truncation. Efficient generation of $Q$ results in high decoding speed.

In order to generate the random integer $Q$ with distribution (\ref%
{GaussDistr}), a naive way is to calculate the cumulative distribution
function
\begin{equation}
F_{c,r}(q)\triangleq P\left( Q\leq q\right) =\sum_{i\leq q}P\left(
Q=i\right) .
\end{equation}%
Obviously, $P(Q=q)=F_{c,r}(q)-F_{c,r}(q-1)$. Therefore, we generate a
real-valued random number $z$ that is uniformly distributed on $[0,1]$; then
we let $Q=q$ if $F_{c,r}(q-1)\leq z<F_{c,r}(q)$. A problem is that this has
to be done online, since $F_{c,r}(q)$ depends on $c$ and $r$. The
implementation complexity can be high, which will slow down decoding.


We now try to find a good approximation to distribution (\ref{GaussDistr}).
Write $r=\lfloor r\rfloor +a$, where $0\leq a<1$. Let $b=1-a$. Distribution (%
\ref{GaussDistr}) can be rewritten as follows%
\begin{equation}
P(Q=q)=\left\{
\begin{array}{l}
e^{-c(a+i)^{2}}/s,\quad q=\lfloor r\rfloor -i \\
e^{-c(b+i)^{2}}/s,\quad q=\lfloor r\rfloor +1+i%
\end{array}%
\right.
\end{equation}%
where $i\geq 0$ is an integer and
\begin{equation*}
s=\sum_{i\geq 0}(e^{-c\left( a+i\right) ^{2}}+e^{-c\left( b+i\right) ^{2}}).
\end{equation*}%
Because $A=\log \rho /\min_{i}\Vert \mathbf{\hat{b}}_{i}\Vert ^{2}$, for every invocation of
Rand\_Round$_{c}\left( r\right) $, we have $c\geq \log \rho $. We use this bound to estimate the
probability $%
P_{2N}$ that $r$ is rounded to the $2N$-integer set $\left\{ \lfloor r\rfloor
-N+1\text{,...,}\lfloor r\rfloor \text{,...,}\lfloor r\rfloor +N\right\} $. Now the probability
that $q$ is not one of these $2N$ points can be bounded as
\begin{eqnarray}
1-P_{2N} &=&\sum_{i\geq N}\left( e^{-c(a+i)^{2}}+e^{-c(b+i)^{2}}/s\right)
\notag \\
&\leq &\left( 1+\rho ^{-(2N+1)}+\rho ^{-(4N+4)}\cdots \right) \cdot  \notag
\\
&&\left( e^{-c(a+N)^{2}}+e^{-c(b+N)^{2}}\right) /s  \notag \\
&<&\left( 1+O(\rho ^{-(2N+1)})\right) \cdot
\notag \\
&&\left( e^{-c(a+N)^{2}}+e^{-c(b+N)^{2}}\right) /s.
\end{eqnarray}%
Here, and throughout this subsection, $O(\cdot)$ is with respect to $N$. Since $s\geq e^{-ca^{2}}$
and $s\geq e^{-cb^{2}}$, we have
\begin{eqnarray}
1-P_{2N} &<&\left( 1+O(\rho ^{-(2N+1)}) \right)
\cdot  \notag  \label{tailbound} \\
&&\left( e^{-c(a+N)^{2}}/e^{-ca^{2}}+e^{-c(b+N)^{2}}/e^{-cb^{2}}\right)
\notag \\
&\leq &2\left( 1+O(\rho ^{-(2N+1)}) \right)
e^{-N^{2}c}  \notag \\
&= &O\left( \rho^{-N^2} \right) .
\end{eqnarray}%
Hence%
\begin{equation}
P_{2N}>1-O(\rho ^{-N^2}) .
\end{equation}%
Since $\rho >1$, the tail bound (\ref{tailbound}) decays very fast. Consequently, it is almost
sure that a call to Rand\_Round$_{c}\left( r\right) $ returns an
integer in $\left\{ \lfloor r\rfloor -N+1\text{,...,}\lfloor r\rfloor \text{%
,...,}\lfloor r\rfloor +N\right\} $ as long as $N$ is not too small.
Therefore, we can approximate distribution (\ref{GaussDistr}) by $2N$-point
discrete distribution as follows.%
\begin{equation}
P(Q=q)=\left\{
\begin{array}{l}
e^{-c(a+N-1)^{2}}/s^{\prime } \\
\quad \vdots \\
e^{-ca^{2}}/s^{\prime } \\
e^{-cb^{2}}/s^{\prime } \\
\quad \vdots \\
e^{-c(b+N-1)^{2}}/s^{\prime }%
\end{array}%
\right.
\begin{array}{l}
q=\lfloor r\rfloor -N+1 \\
\quad \vdots \\
q=\lfloor r\rfloor \\
q=\lfloor r\rfloor +1 \\
\quad \vdots \\
q=\lfloor r\rfloor +N%
\end{array}%
\end{equation}%
where%
\begin{equation*}
s^{\prime }=\sum_{i=0}^{N-1}(e^{-c(a+i)^{2}}+e^{-c(b+i)^{2}}).
\end{equation*}%
%
%
%
%
%
%
%
%
%
%
%
%
%
%
%
%
%
%
%
%
%
%
Fig. 2 shows the distribution (\ref{GaussDistr}), when $r$ $=-5.87$ and $%
c=3.16$. The values of $r$ and $c$ are the interim results obtained by decoding an uncoded
$10\times 10$ system. The distribution of $Q$ concentrates at $%
\lfloor r\rfloor =-6$ and $\lfloor r\rfloor +1=-5$ with probability 0.9 and
0.08 respectively. Fig. 3 compare the bit error rates associated with
different $N$ for an uncoded $10\times 10$ ($n_{T}=n_{R}=10$) system with $%
K=20$. It is seen that the choice of $N=2$ is indistinguishable from larger $%
N$. In fact, it is often adequate to choose a 3-point approximation as the probability in the
central 3 points is almost one.

\begin{figure}[tbp]
\centering
\par
\includegraphics[scale=0.65]{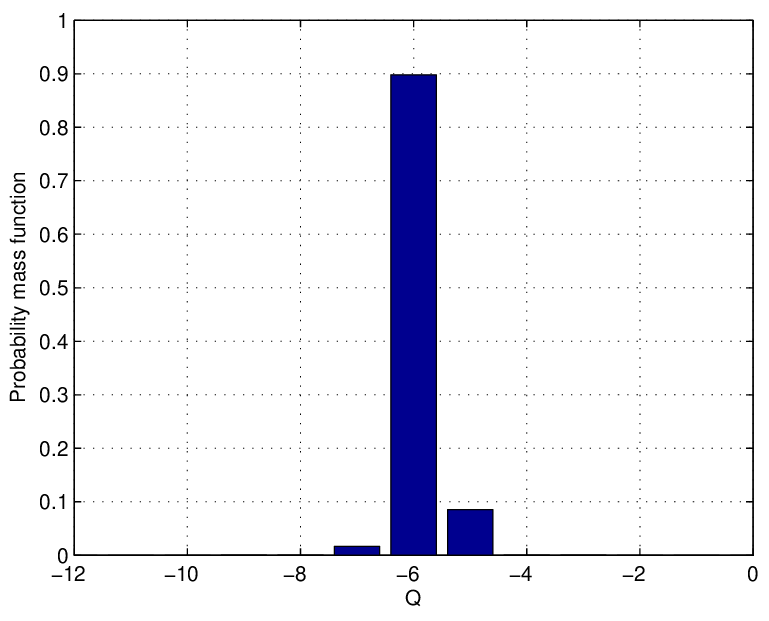}\newline
\caption{Distribution of $Q$ for $r=-5.87$ and $c=3.16$. $P(Q=-7)=0.02$, $%
P(Q=-6)=0.9$ and $P(Q=-5)=0.08$.}
\label{fig:2}
\end{figure}

\begin{figure}[tbp]
\centering
\par
\includegraphics[scale=0.65]{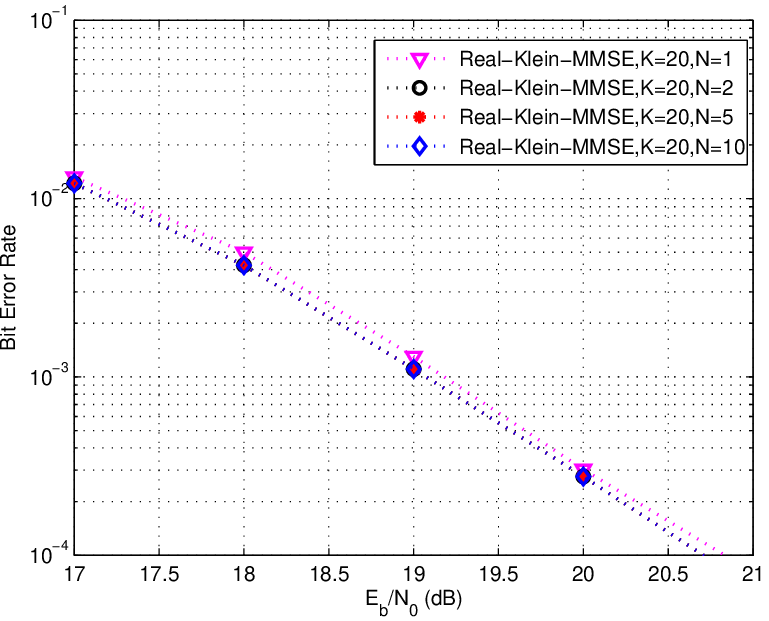}\newline
\caption{Bit error rate vs. average SNR per bit for a $10\times 10$ uncoded
system using 64-QAM.}
\label{fig:3}
\end{figure}

The following lemma provides a theoretical explanation to the above argument from the viewpoint of
\emph{statistical distance} ~\cite[Chap.\ 8]{BK:Micciancio02}. The statistical distance measures
how two probability distributions differ from each other, and is a convenient tool to analyze
randomized algorithms. An important property is that applying a deterministic or random function
to two distributions does not increase the statistical distance. This implies an algorithm behaves
similarly if fed two nearby distributions. More precisely, if the output satisfies a property with
probability~$p$ when the algorithm uses a
distribution~$D_1$, then the property is still satisfied with probability~$%
\geq p - \Delta(D_1,D_2)$ if fed~$D_2$ instead of~$D_1$ (see~\cite[Chap.\ 8]%
{BK:Micciancio02}).

\begin{lemma}
Let~$D$ ($D(i)=P(Q=i)$) be the non-truncated discrete Gaussian distribution,
and~$D^{\prime }$ be the truncated $2N$-point distribution. Then the \emph{%
statistical distance} between~$D$ and~$D^{\prime }$ satisfies:
\begin{equation*}
\Delta(D, D^{\prime }) \triangleq \frac{1}{2} \sum_{i \in \mathbb{Z}} |D(i)-D'(i)| = O(\rho^{
-N^2}).
\end{equation*}
\end{lemma}

\begin{proof}
By definition of~$D^{\prime }$, we have:
\begin{eqnarray*}
\Delta & = & \frac{1}{2} \sum_{i <\lfloor r \rfloor -N + 1} D(i) + \frac{1}{2%
} \sum_{i >\lfloor r \rfloor +N} D(i) \\
&& + \frac{1}{2} \left|1 - \frac{s}{s^{\prime }} \right| \sum_{i = \lfloor r
\rfloor -N + 1}^{\lfloor r \rfloor +N} D(i) \\
& = & \frac{1}{2} \sum_{i <\lfloor r \rfloor -N + 1} D(i) + \frac{1}{2%
} \sum_{i >\lfloor r \rfloor +N} D(i) + \frac{|s^{\prime}-s|}{2s} \\
& \leq & \sum_{i <\lfloor r \rfloor -N + 1} D(i) + \sum_{i >\lfloor r
\rfloor +N} D(i),
\end{eqnarray*}
where~$s = \sum_{i \geq 0}(e^{-c(a+i)^{2}}+e^{-c(b+i)^{2}})$ and~$s^{\prime
}= \sum_{i=0}^{N-1}(e^{-c(a+i)^{2}}+e^{-c(b+i)^{2}})$. The result then
derives from~\eqref{tailbound}.
\end{proof}

As a consequence, the statistical distance between the distributions used by Klein's algorithm
corresponding to the non-truncated and truncated Gaussians is~$nKO(\rho^{-N^2})$. Hence, the
behavior of the algorithm with truncated Gaussian is almost the same.

\subsection{Complex Randomized Lattice Decoding}

Since the traditional lattice formulation is only directly applicable to a real-valued channel
matrix, sampling decoding was given for the real-valued equivalent of the complex-valued channel
matrix. This approach doubles the channel matrix dimension and may lead to higher complexity. From
the complex lattice viewpoint \cite{ComplexLLL}, we study the complex sampling decoding. The
advantage of this algorithm is that it reduces the computational complexity by incorporating
complex LLL reduction \cite{ComplexLLL}.

Due to the orthogonality of real and imaginary part of the complex subchannel, real and imaginary
part of the transmit symbols are decoded in the same step. This allows us to derive complex
sampling decoding by performing randomized rounding for the real and imaginary parts of the
received vector separately.

In this sense, given the real part of input $\mathbf{y}$, sampling
decoding returns real part of $\mathbf{z}$ with probability%
\begin{equation}
P\left( \mathfrak{\Re }\left( \mathbf{z}\right) \right) \geq \frac{1}{%
\prod_{i\leq n}{s(Ar_{i,i}^{2})}}e^{-A\Vert \mathfrak{\Re }\left( \mathbf{y}%
\right) -\mathfrak{\Re }\left( \mathbf{z}\right) \Vert ^{2}}.
\end{equation}%
Similarly, given the imaginary part of input $\mathbf{y}$, sampling
lattice decoding returns imaginary part of $\mathbf{z}$ with probability%
\begin{equation}
P\left( \mathfrak{\Im }\left( \mathbf{z}\right) \right) \geq \frac{1}{%
\prod_{i\leq n}{s(Ar_{i,i}^{2})}}e^{-A\Vert \mathfrak{\Im }\left( \mathbf{y}%
\right) -\mathfrak{\Im }\left( \mathbf{z}\right) \Vert ^{2}}.
\end{equation}%
By multiplying these two probabilities, we get a lower bound on the
probability that the complex sampling decoding returns $\mathbf{z}$%
\begin{eqnarray}
P(\mathbf{z}) &=&P\left( \mathfrak{\Re }\left( \mathbf{z}\right) \right)
\cdot P\left( \mathfrak{\Im }\left( \mathbf{z}\right) \right)  \notag \\
&\geq &\frac{1}{\prod_{i\leq n}{s}^{2}{(Ar_{i,i}^{2})}}e^{-A\left( \Vert
\mathfrak{\Re }\left( \mathbf{y}\right) -\mathfrak{\Re }\left( \mathbf{z}%
\right) \Vert ^{2}+\Vert \mathfrak{\Im }\left( \mathbf{y}\right) -\mathfrak{%
\Im }\left( \mathbf{z}\right) \Vert ^{2}\right) }  \notag \\
&=&\frac{1}{\prod_{i\leq n}{s}^{2}{(Ar_{i,i}^{2})}}e^{-A\Vert \mathbf{y}-%
\mathbf{z}\Vert ^{2}}.
\end{eqnarray}%
Let $A=\log \rho /\min_{i}r_{i,i}^{2}$, where $\rho >1$. Along the same line
of the analysis in the preceding Section, we can easily obtain%
\begin{equation}
P(\mathbf{z})>e^{-\frac{4n}{\rho }(1+{g(\rho )})}\cdot \rho ^{-\Vert \mathbf{%
y}-\mathbf{z}\Vert ^{2}/\min_{1\leq i\leq n}r_{i,i}^{2}}.  \label{Cprob}
\end{equation}%
Given $K$ calls, inequality (\ref{Cprob}) implies the choice of the optimum
value of $\rho $:
\begin{equation}
K=\left( e\rho _{0}\right) ^{4n/\rho _{0}},
\end{equation}%
and the {decoding radius} of complex sampling decoding%
\begin{equation}
R_{\text{Random}}^{\text{C}}=\sqrt{\frac{4n}{\rho _{0}}}%
\min_{1\leq i\leq n}r_{i,i}.
\end{equation}%
Let us compare with the $2n$-dimensional real sampling decoding%
\begin{equation}
R_{\text{Random}}^{\text{R }}=\sqrt{\frac{4n}{\rho _{0}}}%
\min_{1\leq i\leq n}r_{i,i}.
\end{equation}%
Obviously,%
\begin{equation}
R_{\text{Random}}^{\text{C}}=R_{\text{Random}}^{\text{%
R }}
\end{equation}%
Real and complex versions of sampling decoding also have the same parameter $A$ for the same $K$.

\subsection{MMSE-Based Sampling Decoding}

The MMSE detector takes the SNR term into account and thereby leading to an
improved performance. As shown in \cite{wubbenMMSE}, MMSE detector is equal
to ZF with respect to an extended system model. To this end, we define the $%
(m+n)\times n$ extended channel matrix $\underline{\mathbf{B}}$ and the $%
(m+n)\times 1$ extended receive vector $\underline{\mathbf{y}}$ by%
\begin{equation*}
\underline{\mathbf{B}}=\left[
\begin{array}{c}
\mathbf{B} \\
\sigma \mathbf{I}_{n}%
\end{array}%
\right] \text{ \ \ and \ \ }\underline{\mathbf{y}}=\left[
\begin{array}{c}
\mathbf{y} \\
0_{n,1}%
\end{array}%
\right] \text{.}
\end{equation*}%
This viewpoint allows us to incorporate the MMSE criterion in the real and
complex randomized lattice decoding schemes.

\subsection{Soft-Output Decoding}

Soft output is also available from the samples generated in Rand\_SIC.
The $K$ candidate vectors $\mathcal{Z=}\left\{ \mathbf{z}_{1},\cdots,%
\mathbf{z}_{K}\right\} $ can be used to approximate the log-likelihood ratio (LLR), as in
\cite{Hochwald}. For bit $b_i \in \{0,1\}$, the approximated LLR is computed as
\begin{equation}
LLR\left( b_{i}\mid \mathbf{y}\right) =\log \frac{\sum_{\mathbf{z}\mathbf{\in }%
\mathcal{Z}:b_{i}\left( \mathbf{z}\right) =1}\exp \left( -\frac{1}{\sigma
^{2}}\Vert \mathbf{y}-\mathbf{z}\Vert ^{2}\right) }{\sum_{\mathbf{z}\mathbf{\in }%
\mathcal{Z}:b_{i}\left( \mathbf{z}\right) =0}\exp \left( -\frac{1}{\sigma
^{2}}\Vert \mathbf{y}-\mathbf{z}\Vert ^{2}\right) }
\end{equation}%
where $b_{i}\left( \mathbf{z}\right) $ is the $%
i $-th information bit associated with the sample $\mathbf{z}$. The notation $\mathbf{z}%
:b_{i}\left( \mathbf{z}\right) =\mu $ means the set of all vectors $\mathbf{z} $ for which
$b_{i}\left( \mathbf{z}\right) =\mu $.

\subsection{Other issues}

Sampling decoding allows for fully parallel implementation, since the samples can be taken
independently from each other. Thus the decoding speed could be as high as that of a standard
lattice-reduction-aided decoder if it is implemented in parallel.

For SIC, the effective LLL reduction suffices, which has average complexity $O(n^{3}\log n)$ \cite%
{LingISIT07}, and the LLL algorithm can output the matrices $\mathbf{Q}$ and
$\mathbf{R}$ of the QR decomposition.

Since Klein's decoding is random, there is a small chance that all the $K$ samples are further
than the Babai point. Therefore, it is worthwhile always running Babai's algorithm in the very
beginning.
The call can be stopped if the nearest sample point found has distance $\leq
\frac{1}{2}\min_{1\leq i\leq n}r_{i,i}$.

\section{SIMULATION RESULTS}

This section examines the performance of sampling decoding. We assume perfect channel state
information at the receiver. For comparison purposes, the performances of Babai's decoding,
lattice reduction aided
MMSE-SIC decoding, iterative lattice reduction aided MMSE list decoding \cite%
{Shimokawa}, and ML decoding are also shown. Monte Carlo simulation was used to estimate the bit
error rate with Gray mapping.

For LLL reduction, small values of $\delta $ lead to fast convergence, while large values of
$\delta $ lead to a better basis. In our application, increasing $\delta $ will increase the
decoding radius $R_{\text{Random}}$. Since lattice reduction is only performed once at the
beginning of each block, the complexity of LLL reduction is shared by the block. Thus, we set $\delta $%
=0.99 for the best performance. The reduction can be speeded up by applying $\delta=0.75$ to
obtain a reduced basis, then applying $\delta=0.99$ to further reduce it.

Fig. 4 shows the bit error rate for an uncoded system with $n_{T}=n_{R}=10$, $64$-QAM and LLL
reduction ($\delta = 0.99$). Observe that even with $15$ samples (corresponding to a theoretic
gain $G=3$ dB), the performance of the real Klein's decoding enhanced by LLL reduction is
considerably better (by $2.4$ dB) than that of Babai's decoding. Compared to iterative lattice
reduction aided MMSE list decoding with $25$ samples in \cite%
{Shimokawa}, the real Klein's decoding offers not only the improved BER performance (by $1.5$ dB)
but also the promise of smaller list size. MMSE-based real Klein's decoding can achieve further
improvement of $1$ dB. We found that $K=25$ (theoretic gain $G=4$ dB) is sufficient for Real
MMSE-based Klein's
decoding to obtain near-optimum performance for uncoded systems with $%
n_{T}=n_{R}\leq 10$; the SNR loss is less than $0.5$ dB. The complex version of MMSE Klein's
decoding exhibits about $0.2$ dB loss at a BER of $10^{-4}$ when compared to the real version.
Note that the complex LLL algorithm has half of the complexity of real LLL algorithm. At high
dimensions, the real LLL algorithm seems to be slightly better than complex LLL, although their
performances are indistinguishable at low dimensions \cite{ComplexLLL}.

Fig. 5 shows the frame error rate for a $4\times 4$ MIMO system with $4$%
-QAM, using a rate-$1/2$, irregular $\left( 256,128,3\right) $ low-density parity-check (LDPC)
code of codeword length 256 (i.e., 128 information bits). Each codeword spans one channel
realization. The parity check matrix is randomly constructed, but cycles of length $4$ are
eliminated. The maximum number of decoding iterations is set at 50. It is seen that the
soft-output version of sampling decoding is also nearly optimal when $K=24$, with a performance
very close to maximum a posterior probability (MAP) decoding.

Fig. 6 and Fig. 7 show the achieved performance of sampling decoding for the $2\times 2$ Golden
code \cite{GoldenCode} using $16$-QAM and $4\times 4$ Perfect code using $64$-QAM \cite{Oggier}.
The decoding lattices are of dimension 8 and 32 in the real space, respectively. In Fig. 5, the
real MMSE-based Klein decoder with $K=10$ ($G=3$ dB) enjoys 2-dB gain. In Fig. 6, the complex
MMSE-based Klein decoder with $K=20$ ($G=3$ dB), $K=71$ ($G=5$ dB) and $K=174$ ($G=6$ dB) enjoys
3-dB, 4-dB and 5-dB gain respectively. It again confirms that the proposed sampling decoding
considerably narrows the gap to ML performance. Reference \cite{Othman} proposed a decoding scheme
for the Golden code that suffers a loss of $3$ dB with respect to ML decoding, i.e., the
performance is about the same as that of LR-MMSE-SIC. These experimental results are expected, as
LLL reduction has been shown to increase the probability of finding the closest lattice point.
Also, increasing the list size $K$ available to the decoder improves its performance gain. Varying
the number of samples $K$ allows us to negotiate a trade-off between performance and computational
complexity.

Fig. 8 compares the average complexity of 
Babai's decoding, Klein's decoding and sphere decoding for uncoded MIMO systems using $64$-QAM.
The channel matrix remains constant throughout a block of length $100$ and the pre-processing is
only performed once at the beginning of each block. It can be seen that the average flops of
Klein's decoding increases slowly with the dimension, while the average flops of sphere decoding
are exponential in dimension. The computational complexity gap
between Klein's decoding and Babai's decoding is nearly constant for $%
G = 3$ dB or 6 dB. This is because the complexity of Klein's decoding (excluding pre-processing)
is no more than $O(n^3)$ for $G\leq 6$ dB (cf. Table II), meaning the overall complexity is still
$O(n^3\log n)$ (including pre-processing), the same order as that of Babai's decoding.

\begin{figure}[tbp]
\centering
\par
\includegraphics[scale=0.65]{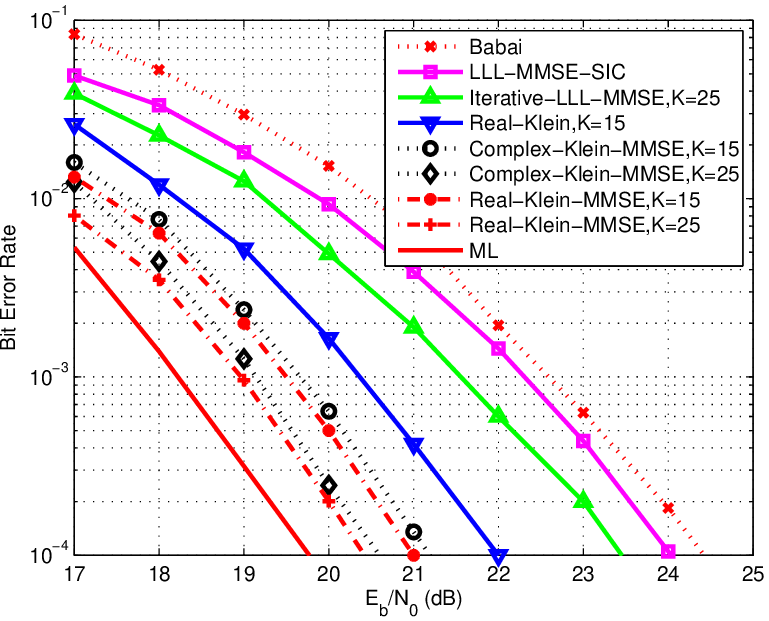}\newline
\caption{Bit error rate vs. average SNR per bit for the uncoded $10 \times
10 $ system using 64-QAM.}
\label{fig:4}
\end{figure}

\begin{figure}[tbp]
\centering
\par
\includegraphics[scale=0.65]{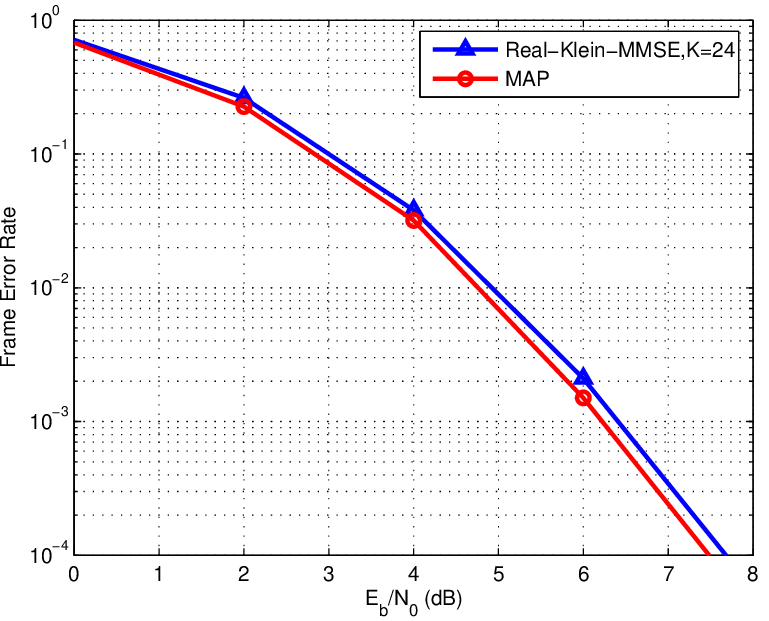}\newline
\caption{Frame error rate vs. average SNR per bit for the $4 \times 4 $ rate-1/2 LDPC code of
codeword length 256 using 4-QAM.} \label{fig:5}
\end{figure}

\begin{figure}[tbp]
\centering
\par
\includegraphics[scale=0.65]{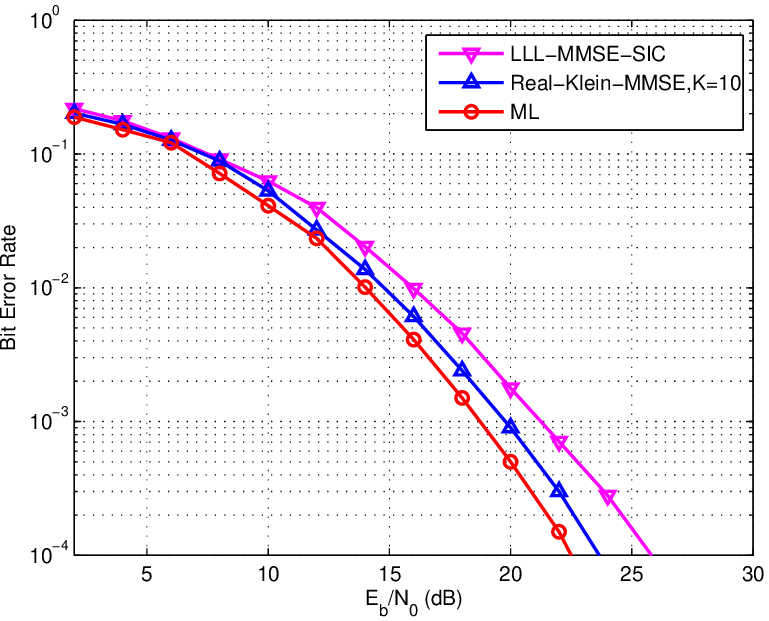}\newline
\caption{Bit error rate vs. average SNR per bit for the $2\times 2$ Golden
code using 16-QAM.}
\label{fig:6}
\end{figure}

\begin{figure}[tbp]
\centering
\par
\includegraphics[scale=0.65]{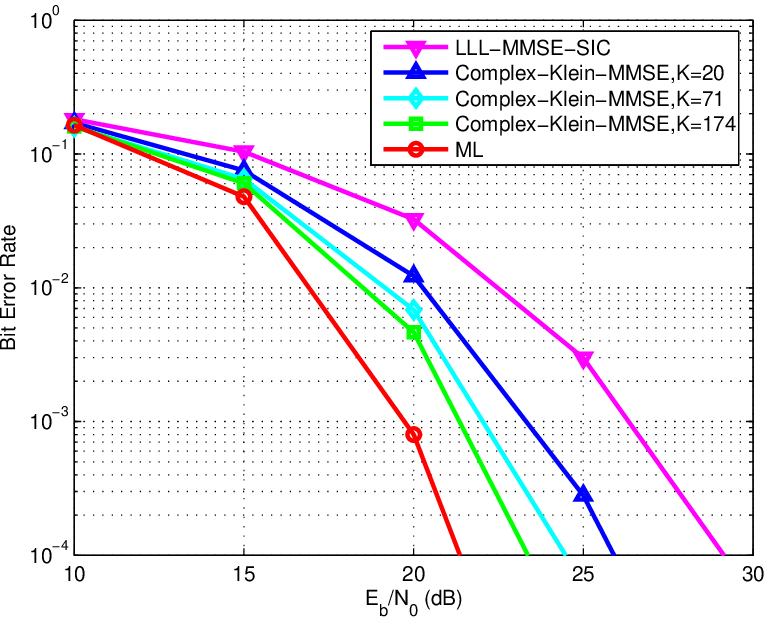}\newline
\caption{Bit error rate vs. average SNR per bit for the $4\times 4$ perfect
code using 64-QAM.}
\label{fig:7}
\end{figure}
\begin{figure}[tbp]
\centering
\par
\includegraphics[scale=0.65]{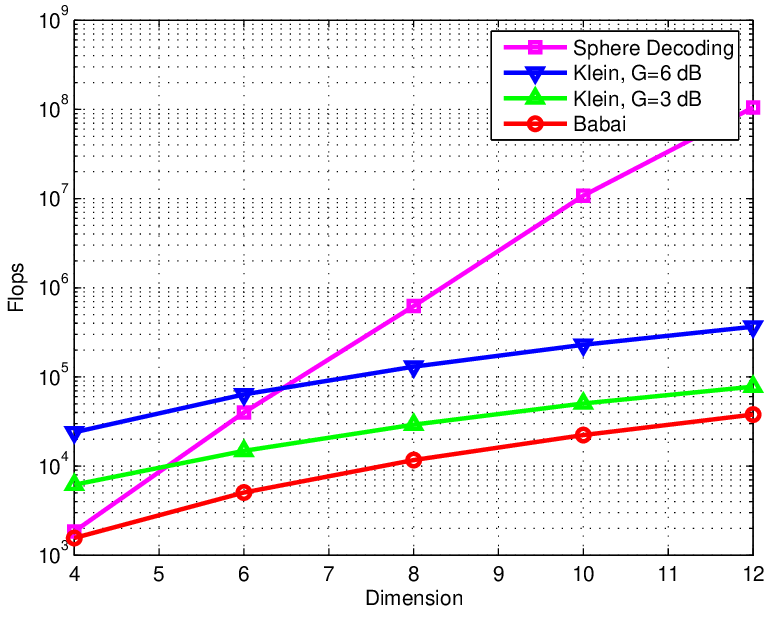}\newline
\caption{Average number of floating-point operations for uncoded MIMO at average SNR per bit = 17
dB. Dimension $n=2n_{T}=2n_{R}$.} \label{fig:8}
\end{figure}

\section{CONCLUSIONS}

In this paper, we studied sampling-based randomized lattice decoding where the standard rounding
in SIC is replaced by random rounding. We refined the analysis of Klein's algorithm and applied it
to uncoded and coded MIMO systems. In essence, Klein's algorithm is a randomized bounded-distance
decoder. Given the number of samples $K$, we derived the
optimum parameter $A$ to maximize the decoding radius $R_{\text{Random}%
}$. Compared to SIC, the best possible gain (measured in squared decoding radius) of our improved
decoder is $G=O(n)$. Although it is asymptotically vanishing compared to the exponential factor of
LLL reduction, the proposed decoder can well be useful in practice. Of particular interest is that
for fixed gain $G$, the value of $K=O(n^{G/4})$ retains the polynomial complexity in $n$. We also
proposed an efficient implementation of random rounding which exhibits indistinguishable
performance, supported by the statistical distance argument for the truncated discrete Gaussian
distribution. The simulations verified that the SNR gain agrees well with $G$ predicted by theory.
With the new approach, it is quite practical to recover $6$ dB of the gap to ML decoding, at
essentially cubic complexity $O(n^3)$. The computational structure of the proposed decoding scheme
is straightforward and allows for an efficient parallel implementation.

\section*{Acknowledgment}
The authors would like to thank the anonymous reviewers for their constructive comments. The third
author gratefully acknowledges the Department of Computing of Macquarie University and the
Department of Mathematics and Statistics of the University of Sydney, where part of this work was
undergone.

\bibliographystyle{IEEEtran}
\bibliography{IEEEabrv,LINGBIB}

\begin{thebibliography}{10}
\providecommand{\url}[1]{#1}
\csname url@samestyle\endcsname
\providecommand{\newblock}{\relax}
\providecommand{\bibinfo}[2]{#2}
\providecommand{\BIBentrySTDinterwordspacing}{\spaceskip=0pt\relax}
\providecommand{\BIBentryALTinterwordstretchfactor}{4}
\providecommand{\BIBentryALTinterwordspacing}{\spaceskip=\fontdimen2\font plus
\BIBentryALTinterwordstretchfactor\fontdimen3\font minus
  \fontdimen4\font\relax}
\providecommand{\BIBforeignlanguage}[2]{{%
\expandafter\ifx\csname l@#1\endcsname\relax
\typeout{** WARNING: IEEEtran.bst: No hyphenation pattern has been}%
\typeout{** loaded for the language `#1'. Using the pattern for}%
\typeout{** the default language instead.}%
\else
\language=\csname l@#1\endcsname
\fi
#2}}
\providecommand{\BIBdecl}{\relax}
\BIBdecl

\bibitem{kannan87}
R.~Kannan, ``Minkowski's convex body theorem and integer programming,''
  \emph{Math. Oper. Res.}, vol.~12, pp. 415--440, Aug. 1987.

\bibitem{Stehle07}
G.~Hanrot and D.~Stehl\'e, ``Improved analysis of {K}annan's shortest vector
  algorithm,'' in \emph{Crypto 2007}, Santa Barbara, California, USA, Aug.
  2007.

\bibitem{MiVo10}
D.~Micciancio and P.~Voulgaris, ``A deterministic single exponential time
  algorithm for most lattice problems based on {V}oronoi cell computations,''
  in \emph{{STOC'10}}, Cambridge, MA, USA, Jun. 2010.

\bibitem{Damen}
M.~O. Damen, H.~E. Gamal, and G.~Caire, ``On maximum likelihood detection and
  the search for the closest lattice point,'' \emph{{IEEE} Trans. Inf. Theory},
  vol.~49, pp. 2389--2402, Oct. 2003.

\bibitem{viterbo}
E.~Viterbo and J.~Boutros, ``A universal lattice code decoder for fading
  channels,'' \emph{{IEEE} Trans. Inf. Theory}, vol.~45, pp. 1639--1642, Jul.
  1999.

\bibitem{agrell}
E.~Agrell, T.~Eriksson, A.~Vardy, and K.~Zeger, ``Closest point search in
  lattices,'' \emph{{IEEE} Trans. Inf. Theory}, vol.~48, pp. 2201--2214, Aug.
  2002.

\bibitem{jalden}
J.~Jald\'{e}n and B.~Ottersen, ``On the complexity of sphere decoding in
  digital communications,'' \emph{{IEEE} Trans. Signal Process.}, vol.~53, pp.
  1474--1484, Apr. 2005.

\bibitem{Oggier}
F.~Oggier, G.~Rekaya, J.-C. Belfiore, and E.~Viterbo, ``Perfect space time
  block codes,'' \emph{{IEEE} Trans. Inf. Theory}, vol.~52, pp. 3885--3902,
  Sep. 2006.

\bibitem{Biglieri09}
E.~Biglieri, Y.~Hong, and E.~Viterbo, ``On fast-decodable space-time block
  codes,'' \emph{{IEEE} Trans. Inf. Theory}, vol.~55, pp. 524--530, Feb. 2009.

\bibitem{Babai}
L.~Babai, ``On {L}ov\'asz' lattice reduction and the nearest lattice point
  problem,'' \emph{Combinatorica}, vol.~6, no.~1, pp. 1--13, 1986.

\bibitem{yao}
H.~Yao and G.~W. Wornell, ``Lattice-reduction-aided detectors for {MIMO}
  communication systems,'' in \emph{Proc. Globecom'02}, Taipei, China, Nov.
  2002, pp. 17--21.

\bibitem{Windpassinger2}
C.~Windpassinger and R.~F.~H. Fischer, ``Low-complexity near-maximum-likelihood
  detection and precoding for {MIMO} systems using lattice reduction,'' in
  \emph{Proc. IEEE Information Theory Workshop}, Paris, France, Mar. 2003, pp.
  345--348.

\bibitem{Taherzadeh:IT}
M.~Taherzadeh, A.~Mobasher, and A.~K. Khandani, ``{LLL} reduction achieves the
  receive diversity in {MIMO} decoding,'' \emph{{IEEE} Trans. Inf. Theory},
  vol.~53, pp. 4801--4805, Dec. 2007.

\bibitem{XiaoliMa08}
X.~Ma and W.~Zhang, ``Performance analysis for {V-BLAST} systems with
  lattice-reduction aided linear equalization,'' \emph{{IEEE} Trans. Commun.},
  vol.~56, pp. 309--318, Feb. 2008.

\bibitem{LingIT07}
\BIBentryALTinterwordspacing
C.~Ling, ``On the proximity factors of lattice reduction-aided decoding,''
  \emph{{IEEE} Trans. Signal Process.}, submitted for publication. [Online].
  Available: \url{http://www.commsp.ee.ic.ac.uk/~cling/}
\BIBentrySTDinterwordspacing

\bibitem{Jalden:LR-MMSE}
J.~Jald\'{e}n and P.~Elia, ``{LR}-aided {MMSE} lattice decoding is {DMT}
  optimal for all approximately universal codes,'' in \emph{Proc. Int. Symp.
  Inform. Theory (ISIT'09)}, Seoul, Korea, 2009.

\bibitem{wubbenMMSE}
D.~W{\"{u}}bben, R.~B{\"{o}}hnke, V.~K{\"{u}}hn, and K.~D. Kammeyer,
  ``Near-maximum-likelihood detection of {MIMO} systems using {MMSE}-based
  lattice reduction,'' in \emph{Proc. {IEEE Int. Conf. Commun.} (ICC'04)},
  Paris, France, Jun. 2004, pp. 798--802.

\bibitem{Klein}
P.~Klein, ``Finding the closest lattice vector when it's unusually close,''
  \emph{Proc. ACM-SIAM Symposium on Discrete Algorithms}, pp. 937--941, 2000.

\bibitem{Othman}
G.~R.-B. Othman, L.~Luzzi, and J.-C. Belfiore, ``Algebraic reduction for the
  {G}olden code,'' in \emph{IEEE Int. Conf. Commun. (ICC'09)}, Dresden,
  Germany, Jun. 2009.

\bibitem{Luzzi10}
L.~Luzzi, G.~R.-B. Othman, and J.-C. Belfiore, ``Augmented lattice reduction
  for {MIMO} decoding,'' \emph{{IEEE} Trans. Wireless Commun.}, vol.~9, pp.
  2853--2859, Sep. 2010.

\bibitem{waters:chase}
D.~W. Waters and J.~R. Barry, ``The {C}hase family of detection algorithms for
  multiple-input multiple-output channels,'' \emph{{IEEE} Trans. Signal
  Process.}, vol.~56, pp. 739--747, Feb. 2008.

\bibitem{Windpassinger1}
C.~Windpassinger, L.~H.-J. Lampe, and R.~F.~H. Fischer, ``From
  lattice-reduction-aided detection towards maximum-likelihood detection in
  {MIMO} systems,'' in \emph{Proc. Int. Conf. Wireless and Optical
  Communications}, Banff, Canada, Jul. 2003.

\bibitem{Karen}
K.~Su and F.~R. Kschischang, ``Coset-based lattice detection for {MIMO}
  systems,'' in \emph{Proc. Int. Symp. Inform. Theory (ISIT'07)}, Jun. 2007,
  pp. 1941--1945.

\bibitem{choi:SIC}
J.~Choi and H.~X. Nguyen, ``{SIC} based detection with list and lattice
  reduction for {MIMO} channels,'' \emph{{IEEE} Trans. Veh. Technol.}, vol.~58,
  pp. 3786--3790, Sep. 2009.

\bibitem{Shimokawa}
T.~Shimokawa and T.~Fujino, ``Iterative lattice reduction aided {MMSE} list
  detection in {MIMO} system,'' in \emph{International Conference on Advanced
  Technologies for Communications {ATC'08}}, Oct. 2008, pp. 50--54.

\bibitem{NJD10}
\BIBentryALTinterwordspacing
H.~Najafi, M.~E.~D. Jafari, and M.~O. Damen, ``On the robustness of lattice
  reduction over correlated fading channels,'' 2010, submitted. [Online].
  Available:
  \url{http://www.ece.uwaterloo.ca/~modamen/submitted/Journal\_TWC.pdf}
\BIBentrySTDinterwordspacing

\bibitem{NgVi08}
P.~Q. Nguyen and T.~Vidick, ``Sieve algorithms for the shortest vector problem
  are practical,'' \emph{J. of Mathematical Cryptology}, vol.~2, no.~2, 2008.

\bibitem{Gentry08}
C.~Gentry, C.~Peikert, and V.~Vaikuntanathan, ``Trapdoors for hard lattices and
  new cryptographic constructions,'' in \emph{40th Annual ACM Symposium on
  Theory of Computing}, Victoria, Canada, 2008, pp. 197--206.

\bibitem{gruber}
P.~M. Gruber and C.~G. Lekkerkerker, \emph{Geometry of Numbers}.\hskip 1em plus
  0.5em minus 0.4em\relax Amsterdam, Netherlands: Elsevier, 1987.

\bibitem{cassels}
J.~W.~S. Cassels, \emph{An Introduction to the Geometry of Numbers}.\hskip 1em
  plus 0.5em minus 0.4em\relax Berlin, Germany: Springer-Verlag, 1971.

\bibitem{Horn}
R.~A. Horn and C.~R. Johnson, \emph{Matrix Analysis}.\hskip 1em plus 0.5em
  minus 0.4em\relax Cambridge, UK: Cambridge University Press, 1985.

\bibitem{LLL}
A.~K. Lenstra, J.~H.~W.~Lenstra, and L.~Lov\'asz, ``Factoring polynomials with
  rational coefficients,'' \emph{Math. Ann.}, vol. 261, pp. 515--534, 1982.

\bibitem{Damien06}
P.~Q. Nguyen and D.~Stehl\'{e}, ``{LLL} on the average,'' in \emph{Proc.
  ANTS-VII}, ser. LNCS 4076.\hskip 1em plus 0.5em minus 0.4em\relax
  Springer-Verlag, 2006, pp. 238--356.

\bibitem{Devroye}
L.~Devroye, \emph{Non-Uniform Random Variate Generation}.\hskip 1em plus 0.5em
  minus 0.4em\relax New York: Springer-Verlag, 1986, pp. 117.

\bibitem{BK:Micciancio02}
D.~Micciancio and S.~Goldwasser, \emph{Complexity of Lattice Problems: A
  Cryptographic Perspective}.\hskip 1em plus 0.5em minus 0.4em\relax Boston:
  Kluwer Academic, 2002.

\bibitem{ComplexLLL}
Y.~H. Gan, C.~Ling, and W.~H. Mow, ``Complex lattice reduction algorithm for
  low-complexity full-diversity {MIMO} detection,'' \emph{{IEEE} Trans. Signal
  Process.}, vol.~57, pp. 2701--2710, Jul. 2009.

\bibitem{Hochwald}
B.~M. Hochwald and S.~ten Brink, ``Achieving near-capacity on a
  multiple-antenna channel,'' \emph{{IEEE} Trans. Commun.}, vol.~51, pp.
  389--399, Mar. 2003.

\bibitem{LingISIT07}
C.~Ling and N.~Howgrave-Graham, ``Effective {LLL} reduction for lattice
  decoding,'' in \emph{Proc. Int. Symp. Inform. Theory (ISIT'07)}, Nice,
  France, Jun. 2007.

\bibitem{GoldenCode}
J.-C. Belfiore, G.~Rekaya, and E.~Viterbo, ``The {G}olden code: A 2 x 2
  full-rate space-time code with nonvanishing determinants,'' \emph{{IEEE}
  Trans. Inf. Theory}, vol.~51, pp. 1432--1436, Apr. 2005.

\end{thebibliography}

\end{document}